\documentclass{article}								
\usepackage{dsfont}
\usepackage{anyfontsize}
\usepackage{xcolor}
\usepackage{graphicx}
\usepackage{aligned-overset}
\usepackage{url}
\usepackage{booktabs}
\usepackage{float}
\usepackage{subcaption}
\usepackage{graphicx}
\usepackage{calc}
\usepackage{amsfonts}
\usepackage{amssymb}
\usepackage{mathtools}
\usepackage{dsfont}
\usepackage{algorithm}
\usepackage[noend]{algpseudocode}
\usepackage{siunitx}
\usepackage{amsthm}
\usepackage{amsmath}
\usepackage{authblk}
\usepackage{geometry}
\usepackage{hyperref}
\newcommand{\RR}{\mathbb{R}}
\newcommand{\NN}{\mathbb{N}}
\newcommand{\PP}{\mathbb{P}}
\newcommand{\EE}{\mathbb{E}}

\DeclareMathOperator*{\argmax}{arg\,max}
\DeclareMathOperator{\vol}{V}
\DeclareMathOperator{\SO}{SO}
\DeclareMathOperator{\interior}{int}
\DeclareMathOperator{\relint}{relint}
\allowdisplaybreaks

\theoremstyle{plain}
\newtheorem{theorem}{Theorem}
\newtheorem{lemma}{Lemma}
\newtheorem{corollary}{Corollary}
\theoremstyle{remark}

\newtheorem{remark}{Remark}
\theoremstyle{definition}
\newtheorem{definition}{Definition} 

\begin{document}

\title{Existence and approximation of densities of chord length- and cross section area distributions}

\author[1]{Thomas van der Jagt}
\author[1]{Geurt Jongbloed}
\author[2]{Martina Vittorietti}
\affil[1]{Delft Institute of Applied Mathematics, Delft University of Technology.}
\affil[2]{Scienze Economiche, Aziendali e Statistiche, Universit\`a degli studi di Palermo.}
\date{}
\maketitle

\begin{abstract}
In various stereological problems an $n$-dimensional convex body is intersected with an $(n-1)$-dimensional Isotropic Uniformly Random (IUR) hyperplane. In this paper the cumulative distribution function associated with the $(n-1)$-dimensional volume of such a random section is studied. This distribution is also known as chord length distribution and cross section area distribution in the planar and spatial case respectively. For various classes of convex bodies it is shown that these distribution functions are absolutely continuous with respect to Lebesgue measure. A Monte Carlo simulation scheme is proposed for approximating the corresponding probability density functions.
\end{abstract}

\section{Introduction} 
In a typical stereological problem we are presented with observations which originate from a lower dimension than the dimension of interest. A classical example is the Wicksell corpuscle problem \cite{Wicksell1925}. The setting of the problem is as follows, balls of varying size are randomly positioned in 3D space. This system of balls is intersected with a plane and the circular section profiles of the balls which happened to be cut by the section plane are observed. The problem is to determine the distribution of the radii of the 3D balls given the distribution of the radii of the observed 2D circular profiles.

An interesting generalization of this problem is to choose a convex shape other than the ball for the shape of the particles. Then, the distribution of observed section areas may be used to estimate the size distribution of the particles. The particles we consider are convex bodies, i.e. compact and convex sets with non-empty interior. In order to deal with such problems we study a class of distributions which is especially important in this setting. Suppose we take some convex body $K \subset \RR^3$ of choice and intersect $K$ with a random section plane. More generally, we may take a convex body $K \subset \RR^n$, and intersect it with a random $(n-1)$-dimensional hyperplane. The random section planes we consider are known as Isotropic Uniformly Random (IUR) planes. This roughly means that every plane which has a non-empty intersection with the convex body has equal probability of occurring. What can be said regarding the cumulative distribution function (CDF) $G_K$ associated with the $(n-1)$-dimensional volume of such a random section of $K$? In this paper we study this kind of distribution functions. In particular, we obtain results on absolute continuity. Whenever we refer to absolute continuity of a cumulative distribution function we mean absolute continuity with respect to Lebesgue measure. We know that the existence and the accurate approximation of the density of $G_K$ is an essential ingredient for defining estimators for particle size distributions in stereological problems. The existence results and the approximation procedure proposed in this paper are used in \cite{vdjagt2023} to design a nonparametric maximum likelihood procedure for estimating particle size distributions.

Given a convex body $K \subset \RR^2$ an IUR section of $K$ is the intersection of $K$ with a random line. The distribution function $G_K$ is then also known as a chord-length distribution function. Some results regarding this function may be found in \cite{Gates1987}. The author notes that it is typically assumed without proof that the CDF of a chord length distribution is absolutely continuous. Only for a limited set of convex polygons there are some results on absolute continuity. See for example \cite{Harutyunyan2009} for the chord length distribution function of a regular polygon which is absolutely continuous. 

For convex bodies $K \subset \RR^3$ the distribution function $G_K$ is sometimes called a cross section area distribution. In \cite{Santalo2004} it is noted that in a stereological setting it is of interest to obtain the density of $G_K$ for some basic shapes such as the simplex or the cube. However, to the best of our knowledge there are no results on whether $G_K$ has a density for a large class of convex bodies, especially in $\RR^n$ with $n \geq 3$. To overcome the difficulty in obtaining an expression for $G_K$, simulations may be used to find an approximation. In \cite{Ohser2000} a description is given for approximating $G_K$ (when $K$ is a polytope in $\RR^3$) and for how it may be used to estimate the size distribution of particles from a sample of observed section areas. 

The outline of this paper is as follows. First, necessary notation and definitions are introduced. Then, we discuss the importance of absolute continuity of $G_K$ for stereological estimation of particle size distributions. This is followed by various results on the distribution function $G_K$. In particular, we show that for a large class of convex bodies, $G_K$ is absolutely continuous. Finally, we propose a Monte Carlo simulation scheme to approximate the corresponding probability density function $g_K$ using density estimation techniques.

\section{Preliminaries}
In this section we introduce the necessary notation and definitions. In particular, we introduce some terminology from convex geometry, a standard reference is \cite{Schneider2013}. In $\RR^n$ a convex body is a convex and compact set with non-empty interior. Let $\mathcal{K}^n$ denote the class of convex bodies in $\RR^n$. Let $\vol_n(K)$ be the $n$-dimensional volume of $K$, its $n$-dimensional Lebesgue measure. $K$ and $L$ indicate convex bodies. Given a point $x \in \RR^n$, the translation of $K$ with $x$ is given by: $K + x =\{k+x:k\in K\}$. The sum of two sets, also known as the Minkowski sum, is defined as: $K + L =\{k+l:k\in K, l \in L\}$. The dilatation or scaling of $K$ with $\lambda > 0$ is given by: $\lambda K =\{\lambda k : k \in K\}$. $\partial K$ denotes the boundary of $K$. Given $x \in \RR^n$ and $r >0$, we write $B(x,r) = \{y \in \RR^n: ||x-y||<r\}$ and $\bar{B}(x,r) = \{y \in \RR^n: ||x-y||\leq r\}$ for the open and closed ball respectively, with radius $r$ centered at $x$. $\SO(n)$ denotes the rotation group of order $n$, containing all orthogonal $n \times n$-matrices of determinant one. Given $M \in \SO(n)$, the rotation of $K$ with $M$ is denoted by: $MK = \{Mk:k\in K\}$. We write $\interior K$ and $\relint K$ to denote the interior and relative interior of $K$ respectively. A convex body $K \in \mathcal{K}^n$ is strictly convex if for all $x,y \in K$ and $\lambda \in (0,1)$ we have $\lambda x + (1-\lambda)y \in \interior K$. A strictly convex body does not have any line segments in its boundary. The unit sphere in $\RR^n$ is denoted by $S^{n-1} = \{(x_1,\dots,x_n) \in \RR^n: x_1^2 + \dots + x_n^2 = 1\}$. The upper hemisphere in $\RR^n$ is given by: $S_{+}^{n-1} = \{(x_1,\dots,x_n) \in S^{n-1}: x_n \geq 0\}$. Let $\sigma_{n-1}$ denote the spherical measure on $S^{n-1}$, also known as the spherical Lebesgue measure on $S^{n-1}$. In integrals over (a subset of) $S^{n-1}$ the notation $\mathrm{d}\theta$ should be interpreted as $\mathrm{d}\sigma_{n-1}(\theta)$. A hyperplane may be parameterized via a unit normal vector $\theta \in S_{+}^{n-1}$ and its signed distance $s\in \RR$ to the origin:
\begin{equation}
    T_{\theta, s} = \{x \in \RR^n: \langle x,\theta\rangle = s\}, \label{hyperplane_eq}
\end{equation}
with $\langle \cdot,\cdot\rangle$ being the usual inner product in $\RR^n$. Given a convex body $K \in \mathcal{K}^n$ its inner section function $m_K : S^{n-1} \to [0,\infty)$ is defined by:
\begin{equation}
m_K(\theta) = \max_{s \in \RR}\vol_{n-1}\left(K \cap T_{\theta, s} \right). \label{inner_section_function}
\end{equation}
This function returns the maximal section volume for any given direction. We can now define what is meant by an Isotropic, Uniformly Random (IUR) plane hitting $K$. The following definition gives a convenient parameterization of IUR planes, see \cite{Baddeley2004} for IUR plane sections of convex bodies in $\RR^3$ (the generalization to $\RR^n$ is straightforward):

\begin{definition}[IUR plane]\label{iur_plane_def}
An IUR plane $T$ hitting a fixed $K \in \mathcal{K}^n$, $n \geq 2$, is defined as $T = T_{\Theta, S}$ where $(\Theta, S)$ has joint probability density, $f_K : S_{+}^{n-1} \times \RR \to [0, \infty)$ given by:
\begin{equation*}
    f_K(\theta, s) = \begin{cases}\frac{1}{\mu([K])} & \text{if } K \cap T_{\theta,s} \neq \emptyset \\
    0 & \text{otherwise,}
    \end{cases}
\end{equation*}
with $T_{\theta,s}$ as in (\ref{hyperplane_eq}) and
\begin{equation}
    \mu([K]) = \int_{S_{+}^{n-1}}\int_{-\infty}^\infty \mathds{1}{\{K \cap T_{\theta,s} \neq \emptyset\}}\mathrm{d}s\mathrm{d}\theta. \label{normalization_constant}
\end{equation}
\end{definition}

The notion of IUR planes was originally introduced in \cite{Davy1977}. It is important to highlight that there are other kinds of random planes which appear in stereological problems, hence care should be taken in considering the appropriate distribution. See \cite{Miles1978} for more details. Note that the distribution in Definition \ref{iur_plane_def} is a joint uniform distribution; the marginals are in general not uniform. We stress that the density $f_K$ prescribes the probability associated with the possible locations and orientations of the section plane, not the volumes of hyperplane sections. Fix $K \in \mathcal{K}^n$ and let $f_K$ be as in Definition \ref{iur_plane_def}. Integrating out the variable $s$, we obtain the marginal density:
\begin{equation}
f_{K,\Theta}(\theta) = \frac{L(p_{\theta}(K))}{\mu([K])},  \text{\quad}  \theta \in S_{+}^{n-1}.\label{direction_marginal_density}
\end{equation}
In (\ref{direction_marginal_density}), $p_\theta(K)$ represents the orthogonal projection of $K$ on the line through the origin with direction $\theta$. $L(p_\theta(K))$ is then the length of this orthogonal projection, hence $L(p_\theta(K))$ may also be called the width of $K$ in direction $\theta$. The constant $\mu([K])$ is related to the average width $\bar{b}(K)$, via:
\begin{equation}
    \mu([K]) = \alpha_n \bar{b}(K).
\end{equation}
The average width is defined as:
\[\bar{b}(K) = \frac{1}{\alpha_n}\int_{S_{+}^{n-1}} L(p_{\theta}(K))\mathrm{d}\theta,\]
and the constant $\alpha_n$ is given by:
\begin{equation*}
    \alpha_n = \sigma_{n-1}\left(S_{+}^{n-1}\right) = \frac{\pi^{\frac{n}{2}}}{\Gamma\left(\frac{n}{2}\right)}.
\end{equation*}
Conditioning an IUR plane on a fixed direction yields a so-called Fixed orientation Uniformly Random (FUR) plane. Fix $\theta \in S^{n-1}$, let $a=a(\theta)$ be the smallest number such that $K\cap T_{\theta, a} \neq \emptyset$, similarly let $b=b(\theta)$ be the largest number such that $K\cap T_{\theta, b} \neq \emptyset$. Then, conditional on this direction $\Theta = \theta$, $S$ is uniformly distributed on the interval $[a, b]$, and we denote this conditional density by:
\begin{equation}
f_{S|\Theta}(s|\theta) = \begin{cases}\frac{1}{b(\theta) - a(\theta)} & \text{if } s \in [a(\theta), b(\theta)] \\
    0 & \text{otherwise,}
    \end{cases} \label{diretion_conditional_density}
\end{equation}
We may also write: $S|\Theta = \theta \sim \mathcal{U}(a(\theta),b(\theta))$. Given a cumulative distribution function (CDF) $F$ or a probability density function (PDF) $f$, we write $X\sim F$ or $X \sim f$ to indicate that the random variable $X$ is distributed according to $F$ or $f$ respectively. We are now ready to introduce the CDF of interest in this paper.

\begin{definition}[section volume CDF]\label{section_cdf_def}
Fix $K \in \mathcal{K}^n$, $n \geq 2$, let $f_K$ be as in Definition \ref{iur_plane_def}. Let $(\Theta, S) \sim f_K$, the random variable $Z = \vol_{n-1}(K\cap T_{\Theta, S})$ has cumulative distribution function $G_K$ which is given by:
\begin{align*}
    G_K(z) = \int_{S_{+}^{n-1}}\int_{\RR} \mathds{1}{\{\vol_{n-1}(K\cap T_{\theta, s}) \leq z\}}f_K(\theta, s)\mathrm{d}s\mathrm{d}\theta.
\end{align*}
We refer to $G_K$ as the section volume CDF of $K$.
\end{definition}
We remark that the expression of the CDF $G_K$ follows from the fact that $G_K(z) = P(Z \leq z) = \EE(\mathds{1}{\{Z \leq z\}})$ and the law of the unconscious statistician. In $\RR^2$ we may still refer to $G_K$ as chord length distribution function and in $\RR^3$ we may call it cross section area distribution function. 

\section{Stereological estimation of size distributions}
In this section we show how absolute continuity of $G_K$ and accurate approximation of its density $g_K$ is important for stereological estimation of particle size distributions. We consider the generalization of the Wicksell corpuscle problem as mentioned in the introduction. Suppose we pick some particle, a convex body $K \subset \RR^3$, and instances of $K$ of varying size are randomly positioned and oriented in $\RR^3$. Such an isotropic system of particles is often described in terms of a germ-grain model, see for example \cite{Ohser2000} and sections 6.5 and 10.5 in \cite{Chiu2013}. In this setting an isotropic typical particle is chosen and the particles are positioned in $\RR^3$ using a stationary point process. The particles have random sizes and a particle of size $\lambda$ is equal to $\lambda K$ up to rotation and translation. If the diameter of $K$ equals 1, then the size $\lambda$  of a particle is simply its diameter. The sizes of the particles are independent and identically distributed according to the distribution function $H$. The mean particle size is given by:
\[\EE(\Lambda) = \int_0^\infty \lambda\mathrm{d}H(\lambda).\]
Intersecting this system of particles with a plane, the distribution function associated with the area of a typical section profile is denoted by $F_A$. It can be shown, see for example \cite{vdjagt2023}, that $F_A$ is given by:
\begin{equation*}
F_A(a) = \frac{1}{\EE(\Lambda)}\int_0^\infty G_{K}\left(\frac{a}{\lambda^2}\right)\lambda\mathrm{d}H(\lambda). 
\end{equation*}
Let $a_{\text{max}}$ denote the largest possible section area over all planar sections of $K$. If $G_K$ is absolutely continuous and has Lebesgue density $g_K$, supported on $(0, a_{\text{max}})$, then $F_A$ has a density given by:
\begin{equation}
    f_A(a) = \frac{1}{\EE(\Lambda)}\int_{\sqrt{\frac{a}{a_{\text{max}}}}}^\infty g_{K}\left(\frac{a}{\lambda^2}\right)\frac{1}{\lambda}\mathrm{d}H(\lambda).\label{density_stereological_equation}
\end{equation}
Another derivation of (\ref{density_stereological_equation}) appears in chapter 16 of \cite{Santalo2004}. The main implication of (\ref{density_stereological_equation}) is that given a sample of observed section areas, corresponding to some system of particles, the likelihood is well-defined. This means that the size distribution $H$ may be estimated using likelihood-based methods of statistical inference. For any given candidate $H'$ for $H$ evaluating such a likelihood requires $g_K$ to be known. This density is in general hard to compute and one way to deal with this is to use the density approximation procedure presented in this paper, which can approximate it arbitrarily closely.

Suppose that we obtain a sample $A_1,\dots,A_N$ which is independent and identically distributed according to $f_A$. We wish to estimate the distribution function $H$ of the size distribution. This problem is identifiable, in \cite{vdjagt2023} it is shown that the profile area distribution $f_A$ uniquely determines $H$ under lenient assumptions. Moreover, the authors also derive a non-parametric estimator for the so-called length-biased size distribution via non-parametric maximum-likelihood, and show that it is consistent. We now briefly discuss the definition of this estimator, for further details we refer to  \cite{vdjagt2023}. The length-biased size distribution, or length-biased version of $H$ is given by:
\[H^b(\lambda):=\frac{\int_0^\lambda x\mathrm{d}H(x)}{\int_0^\infty x\mathrm{d}H(x)}.\]
As in the Wicksell corpuscle problem we are dealing with length-biased sampling, meaning that the probability that a particle is hit by the section plane is proportional to its size. Hence, while the size of a typical particle is distributed according to $H$, the size of a typical particle in the section plane is distributed according to $H^b$. Letting $A \sim f_A$, set $S =\sqrt{A}$ and let $f_S$ denote the density of $S$. Analogously, let $Z \sim g_K$, and let $g_K^S$ denote the density of $\sqrt{Z}$. Plugging the definitions of these densities into (\ref{density_stereological_equation}) yields:
\begin{equation}
    f_S(s) = \int_{0}^\infty g_K^S\left(\frac{s}{\lambda}\right)\frac{1}{\lambda}\mathrm{d}H^b(\lambda). \label{fs_density}
\end{equation}
Because $g_K^S$ is supported on $(0,\sqrt{a_{\text{max}}})$, in (\ref{fs_density}) the lower bound of the integration region is effectively $s/\sqrt{a_{\text{max}}}$ instead of $0$. Set $S_i = \sqrt{A_i}$ for $i \in \{1,\dots,N\}$, and let $s_1 < s_2 <\dots < s_N$ be a realization of the order statistics of $S_1,\dots,S_N$. The estimator $\hat{H}_N^b$ for $H^b$ is defined as a maximizer of the (scaled by $1/N$) log-likelihood:
\begin{align*}
    \hat{H}_N^b \in \argmax_{H^b \in \mathcal{F}_N^{+}} \frac{1}{N}\sum_{i=1}^N \log\left(\int_0^\infty g_K^S\left(\frac{s_i}{\lambda}\right)\frac{1}{\lambda}\mathrm{d}H^b(\lambda) \right). \label{hb_mle}
\end{align*}
Here, we maximize over $\mathcal{F}_N^{+}$, the class of all piece-wise constant distribution functions on $(0, \infty)$, with jump-locations restricted to the set of observations, the $s_i$'s. A single realization of $\hat{H}_N^b$ is shown in Figure \ref{mle_example_fig}. 

For the simulation result in Figure \ref{mle_example_fig}, each particle is a convex dodecahedron and the underlying size distribution is a standard exponential distribution. The corresponding length-biased distribution $H^b$ is a gamma distribution. A sample of size $N=1000$ from $f_S$ is used to compute $\hat{H}_N^b$. We note that for some applications an estimate of $H^b$ may be sufficient. If an estimate of $H$ is desired, a procedure which uses $\hat{H}_N^b$ to obtain an estimate of $H$ may be found in \cite{vdjagt2023}.

\begin{figure}[t]
    \centering
    \includegraphics{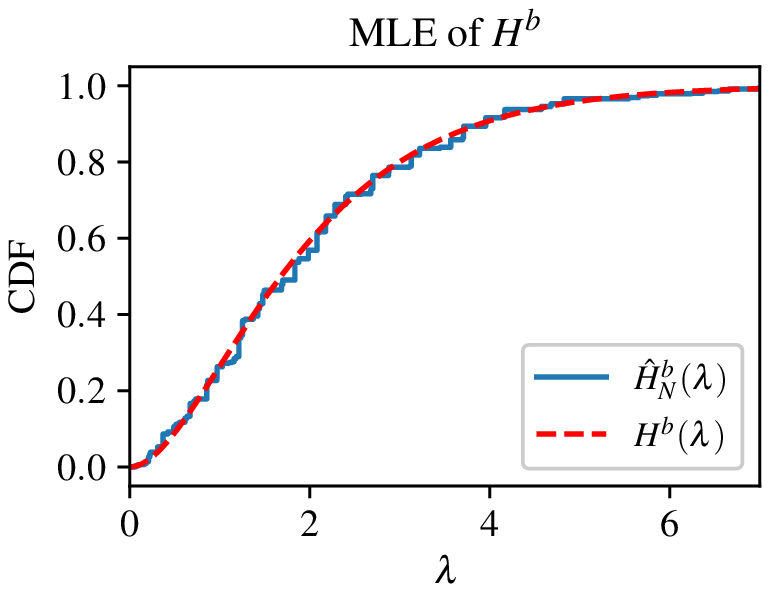}
    \caption{A single realization of the maximum likelihood estimator $\hat{H}_N^b$ ($N=1000$) and the true length-biased size distribution $H^b$.}
    \label{mle_example_fig}
\end{figure}

\section{Properties of the section volume CDF}

In this section we derive various properties of the section volume CDF as described in Definition \ref{section_cdf_def}. Given a convex body $K$ the following Lemma highlights some basic properties.

\begin{lemma}\label{basic_properties}
Fix $K, L \in \mathcal{K}^n$, let $G_K$, $G_L$ be their section volume CDF respectively. Let $z \in \RR$, then:
\begin{enumerate}
    \item Translation invariance: $G_{K+x}(z) = G_K(z)$ for all $x \in \RR^n$.
    \item Rotation invariance: $G_{MK}(z) = G_K(z)$ for all $M \in \SO(n)$.
    \item Scaling: $G_{\lambda K}(z) = G_K\left(z/\lambda^{n-1}\right)$ for all $\lambda > 0$.
    \item Inclusion: If $K \subset L$ then:
        \[G_L(z) \leq G_K(z)\frac{\bar{b}(K)}{\bar{b}(L)} + \left(1 - \frac{\bar{b}(K)}{\bar{b}(L)} \right).\]
\end{enumerate}
\end{lemma}
The translation and rotation invariance of IUR planes is a defining property of IUR planes, \cite{Davy1977}, and it may be used to prove property 1 and 2 in Lemma \ref{basic_properties}. The third property also appears in \cite{Santalo2004} for $n=3$. All of these properties are well-known for chord length distributions and the generalization to $(n-1)$-dimensional sections of convex bodies in $\RR^n$ is not difficult. For the sake of completeness, the proof of this Lemma may be found in the Appendix at the end of this paper.

We need Brunn's theorem (see for example \cite{Koldobsky2005}) to prove one of our main results:

\begin{theorem}[Brunn]\label{Brunn_thm}
Let $K \in \mathcal{K}^n$, $n \geq 2$. Fix $\theta \in S^{n-1}$. The function $f_{\theta} : \RR \to [0, \infty)$ given by:
\[f_{\theta}(s) = \vol_{n-1}(K \cap T_{\theta, s})^{\frac{1}{n-1}}\]
is concave on its support.
\end{theorem}
Ignoring the exponent $1/(n-1)$ in the definition of $f_{\theta}$, this function returns the volume of the intersection of $K$ with $T_{\theta, s}$. Because we fix $\theta$ this means the function considers volumes of parallel slices of $K$, and it is a function of the (signed) distance of the section plane to the origin. The statement of Brunn's Theorem inspires us to study a distribution function which is closely related to $G_K$:

\begin{definition}[Transformed section volume CDF]
Fix $K \in \mathcal{K}^n$, $n \geq 2$, let $f_K$ be as in Definition \ref{iur_plane_def}. Let $(\Theta, S) \sim f_K$, the random variable $Z = \vol_{n-1}(K\cap T_{\Theta, S})^{1/(n-1)}$ has cumulative distribution function $G_K^S$ which is given by:
\begin{align*}
    G_K^S&(z) = \int_{S_{+}^{n-1}}\int_\RR \mathds{1}{\{\vol_{n-1}(K\cap T_{\theta, s})^{\frac{1}{n-1}} \leq z\}}f_K(\theta, s)\mathrm{d}s\mathrm{d}\theta.
\end{align*}
We refer to $G_K^S$ as the transformed section volume CDF of $K$.
\end{definition}

This distribution function $G_K^S$ turns out to be more natural to study compared to $G_K$. This will become clear in the proof of the upcoming theorem. Note that, for $K \in \mathcal{K}^n$, $G_K$ and $G_K^S$ are related as follows:
\begin{equation*}
G_K^S(z) = G_K\left(z^{n-1}\right).
\end{equation*}
\begin{remark}\label{remark_absolute_continuity}
$G_K$ is absolutely continuous if and only if $G_K^S$ is absolutely continuous. After all, suppose that $G_K^S$ has probability density function $g_K^S$. Let $X \sim g_K^S$, then $X^{n-1} \sim G_K$ and via the well-known change of variables formula this random variable has the following probability density function:
\begin{equation}
g_K(z) = g_K^S\left(z^{\frac{1}{n-1}}\right)\frac{z^{\frac{2-n}{n-1}}}{n-1}. \label{g_and_h_relation}
\end{equation}
The converse case is analogous. 
\end{remark}

We now present one of the main theorems in this paper:

\begin{theorem}\label{strict_concavity_condition}
Let $K \in \mathcal{K}^n$, $n \geq 2$. Define the function $f_{\theta} : \RR \to [0, \infty)$ by:
\[f_{\theta}(s) = \vol_{n-1}(K \cap T_{\theta, s})^{\frac{1}{n-1}}.\]
If $f_{\theta}$ has a unique maximum and is continuous on $\RR$ for almost all $\theta \in S_{+}^{n-1}$, then $G_K$ is absolutely continuous with respect to Lebesgue measure.
\end{theorem}

\begin{proof}
Given $K \in \mathcal{K}^n$, let $G_K$ be its section volume CDF and let $G_K^S$ be its transformed section volume CDF. We show that $G_K^S$ is absolutely continuous, from this it follows that $G_K$ is absolutely continuous by Remark \ref{remark_absolute_continuity}. By conditioning the distribution function $G_K^S$ on $\Theta$ having a particular value, $G_K^S$ may be written as a mixture distribution:
\begin{align*}
G_K^S(z) = \PP\left(f_{\Theta}(S) \leq z \right) = \int_{S_{+}^{n-1}}\PP\left(f_{\theta}(S) \leq z \Big| \Theta = \theta \right)f_{K,\Theta}(\theta)\mathrm{d}\theta,
\end{align*}
with $f_{K,\Theta}(\theta)$ being the marginal density of $\Theta$ as in (\ref{direction_marginal_density}) and $f_{\theta}(\cdot)$ as in the statement of the theorem. For notation convenience, write:
\begin{equation}
G_K^S(z|\theta) := \PP\left(f_{\theta}(S) \leq z \Big| \Theta = \theta \right).
\end{equation}
Let $a=a(\theta)$ and $b=b(\theta)$ be as in (\ref{diretion_conditional_density}) such that $S|\Theta = \theta \sim \mathcal{U}(a,b)$. Choose $\theta \in S_{+}^{n-1}$ such that $f_{\theta}$ has a unique maximum and is continuous on $\RR$. By definition of $a$ we know that $T_{\theta,a}$ intersects $K$ only through the boundary of $K$. By the assumed continuity of $f_{\theta}$ we have $f_{\theta}(a) = 0$ and similarly: $f_\theta(b) = 0$. Note that $f_\theta$ has the following domain and codomain:
\begin{equation}
f_{\theta}:[a,b] \to D_{\theta}, \text{\ with \ } D_{\theta} = \left[0, m_K(\theta)^{\frac{1}{n-1}}\right], \label{domain_codomain}
\end{equation}
and $m_K(\cdot)$ as in (\ref{inner_section_function}). By Brunn's theorem $f_\theta$ is concave on its support and by assumption it attains its maximum in a single point $c$. As a result, $f_\theta$ is strictly increasing on $(a,c)$ and strictly decreasing on $(c,b)$. Therefore, $f_{\theta}$ restricted to $(a,c)$ is invertible, and its inverse is convex and strictly increasing. Let:
\[f_{\theta}^{+}:\left(0,m_K(\theta)^{\frac{1}{n-1}}\right) \to (a, c),\]
denote this inverse. Similarly, $f_{\theta}$ restricted to $(c,b)$ has an inverse:
\[f_{\theta}^{-}:\left(0,m_K(\theta)^{\frac{1}{n-1}}\right) \to (c, b),\] 
which is concave and strictly decreasing. Write:
\[p := \PP(S \in (a,c)|\Theta=\theta).\]
By using the fact that $S|\Theta = \theta \sim \mathcal{U}(a,b)$ we find $p=(c-a)/(b-a)$. Moreover, we obtain the following expression for $G_K^S(z|\theta)$:
\begin{align*}
G_K^S(z|\theta) &= \PP\left(f_{\theta}(S) \leq z \Big| \Theta = \theta, S\in(a,c) \right)p + \PP\left(f_{\theta}(S) \leq z \Big| \Theta = \theta, S\in(c,b) \right)(1-p)\\
&= \PP\left(S \leq f_{\theta}^{+}(z) \Big| \Theta = \theta, S\in(a,c) \right)p + \PP\left(S \geq f_{\theta}^{-}(z) \Big| \Theta = \theta, S\in(c,b) \right)(1-p)\\
&= \frac{f_{\theta}^{+}(z) - a}{c - a}p + \left(1 - \frac{f_{\theta}^{-}(z) - c}{b - c}\right)(1-p).
\end{align*}
Because $f_{\theta}^{-}$ is concave and strictly decreasing, $-f_{\theta}^{-}$ is convex and strictly increasing. Therefore, $G_K^S(\cdot|\theta)$ is a convex combination of two functions both of which are convex and strictly increasing on the interval $D_{\theta}$ (as in (\ref{domain_codomain})). As a result, $G_K^S(\cdot|\theta)$ is convex, continuous, and strictly increasing on $D_{\theta}$, which is the support of this distribution function. We conclude that for almost all $\theta \in S_{+}^{n-1}$, $G_K^S(\cdot|\theta)$ is absolutely continuous because it is convex on its support and continuous on $\RR$. Finally, this means that $G_K^S$ as a mixture of absolutely continuous distribution functions is absolutely continuous by Fubini's theorem.
\end{proof}

\begin{remark}\label{convexity_hk_remark}
The arguments used in the proof of Theorem \ref{strictly_convex_result} do not hold for general convex bodies. For general convex bodies the function $f_{\theta}$ is concave by Brunn's theorem. Therefore, the set of points at which it attains its maximum may be an interval rather than a single point. When this is the case, $G_K^S(\cdot|\theta)$ is still convex on its support, but it is discontinuous in the point $m_K(\theta)^{1/(n-1)}$, which is the right boundary point of its support. As a result, for any convex body $K \in \mathcal{K}^n$, $G_K^S$ is convex on the interval $(0,d_K)$ with: $d_K = \min_{\theta \in S^{n-1}}m_K(\theta)^{1/(n-1)}$.
\end{remark}

\subsection{Strictly convex bodies}

Let us now consider a particular class of convex bodies known as strictly convex bodies. The class of strictly convex bodies is large in a precise sense. For one, the class of convex bodies which are not smooth or strictly convex form a set of first Baire category, see \cite{Zamfirescu1987} for details. We have not yet mentioned smooth convex bodies, loosely speaking it means that their boundary is smooth. An important result we obtain in this section is that given that $K \in \mathcal{K}^n$ is strictly convex, then $G_K$ is absolutely continuous. Therefore, we show that for a large class of convex bodies $G_K$ is absolutely continuous. The main tools to obtain this result are the famous Brunn-Minkowski inequality and a variant of Brunn's theorem. In the field of convex geometry the importance of the Brunn-Minkowski inequality cannot be overstated, we refer to the review paper \cite{Gardner2002} for variants of the theorem and its applications.

\begin{theorem}[Brunn-Minkowski]\label{brunnmin}
Given convex bodies $K, L \in \mathcal{K}^n$ and $0 < \lambda < 1$ the following inequality holds:
\begin{equation*}
    \vol_{n}(\lambda K + (1-\lambda) L)^{\frac{1}{n}} \geq \lambda\vol_{n}(K)^{\frac{1}{n}} + (1-\lambda)\vol_{n}(L)^{\frac{1}{n}},
\end{equation*}
with equality if and only if $K$ and $L$ are equal up to translation and dilatation.
\end{theorem}
The equality condition in Theorem \ref{brunnmin} means that there exist $\delta > 0$ and $x \in \RR^n$ such that $K = \delta L + x$. In order to prove that $G_K$ is absolutely continuous for strictly convex $K \in \mathcal{K}^n$ we show that the conditions in Theorem \ref{strict_concavity_condition} are satisfied. First, we need the following Lemma:

\begin{lemma}\label{strict_inclusion_lemma}
Let $K, L \in \mathcal{K}^n$ with $K \subset \interior L$, then $\vol_{n}(K) < \vol_n(L)$.
\end{lemma}

Its proof is given in the Appendix at the end of this paper. We show that the strict convexity of a convex body carries over to strict concavity of the function $f_\theta$ (as in Theorem \ref{Brunn_thm}).

\begin{theorem}\label{Brunn_strict_thm}
Let $K \in \mathcal{K}^n$ be a strictly convex body, $n \geq 2$. Fix $\theta \in S^{n-1}$. The function $f_{\theta} : \RR \to [0, \infty)$ given by:
\[f_{\theta}(s) = \vol_{n-1}(K \cap T_{\theta, s})^{\frac{1}{n-1}}\]
is continuous on $\RR$ and strictly concave on its support.
\end{theorem}

\begin{proof}
The proof is a slight variation of a proof of Brunn's theorem using the Brunn-Minkowski inequality as found in \cite{Koldobsky2005} (pp 18, 19). Fix $\theta \in S^{n-1}$. Choose $r,t$ in the support of $f_{\theta}$, such that $r < t$. Let $\lambda \in (0,1)$, set $s = \lambda r + (1-\lambda)t$ and consider the hyperplane sections $K_r := K \cap T_{\theta, r}$, $K_s := K \cap T_{\theta, s}$ and $K_t := K \cap T_{\theta, t}$. We show that:
\begin{equation}
\lambda K_r + (1-\lambda) K_t \subset (\interior K) \cap T_{\theta, s}. \label{strict_concavity_proof_1}
\end{equation}
Let $z \in \lambda K_r + (1-\lambda) K_t$, then $z = \lambda x + (1-\lambda)y$ for some $x \in K_r$ and some $y \in K_t$.
Because also $x,y \in K$ we have $z \in \interior K$ due to the strict convexity of $K$. Also, note that $\langle z, \theta\rangle = \lambda \langle x , \theta\rangle + (1-\lambda)\langle y , \theta\rangle = \lambda r + (1-\lambda)t = s$. Hence, $z \in T_{\theta, s}$, which proves (\ref{strict_concavity_proof_1}). It can readily be verified that: $(\interior K) \cap T_{\theta, s} = \relint \left(K_s\right)$. Combining this with (\ref{strict_concavity_proof_1}) we find: $\lambda K_r + (1-\lambda) K_t \subset \relint \left(K_s\right)$. Let $\Pi(L)$ denote the orthogonal projection of $L$ on the hyperplane $T_{\theta, 0}$. Note that $\lambda K_r + (1-\lambda) K_t$ and $K_s$ are subsets of $T_{\theta, s}$, projecting them on $T_{\theta, 0}$ preserves the inclusion:
\[\Pi(\lambda K_r + (1-\lambda) K_t) \subset \relint\Pi(K_s).\]
Identifying $T_{\theta, 0}$ with $\RR^{n-1}$ we may regard $\Pi(\lambda K_r + (1-\lambda) K_t)$ and $\Pi(K_s)$ as convex bodies in $\RR^{n-1}$. Under this identification $\Pi(\lambda K_r + (1-\lambda) K_t) \subset \interior \Pi(K_s)$ and applying Lemma \ref{strict_inclusion_lemma} yields:
\begin{equation*}
\vol_{n-1}\left(\Pi(\lambda K_r + (1-\lambda) K_t)\right) < \vol_{n-1}\left(\Pi(K_s)\right).
\end{equation*}
Keep in mind that projecting a set on $T_{\theta, 0}$ does not affect its $(n-1)$-dimensional volume. Note that we may change the order of these projections and the Minkowski sum:
\begin{equation*}
\lambda \Pi(K_r) + (1-\lambda) \Pi(K_t) = \Pi(\lambda K_r + (1-\lambda) K_t),  
\end{equation*}
where the sum of sets is considered in the plane $T_{\theta, 0}$. Hence,
\begin{equation}
    \vol_{n-1}\left(\Pi(K_s)\right) > \vol_{n-1}(\Pi(\lambda K_r + (1-\lambda) K_t)) = \vol_{n-1}(\lambda \Pi(K_r) + (1-\lambda) \Pi(K_t)). \label{strict_concavity_proof_2}
\end{equation}
Once again, $\Pi(K_r)$ and $\Pi(K_t)$ may be identified as convex bodies in $\RR^{n-1}$ and we may apply Brunn-Minkowski's (B.M.) inequality to obtain the desired result:
\begin{align*}
    f_{\theta}(s) &= \vol_{n-1}(K_s)^{\frac{1}{n-1}}  \\
     &= \vol_{n-1}(\Pi(K_s))^{\frac{1}{n-1}}  \\
     \overset{(\ref{strict_concavity_proof_2})}&{>} \vol_{n-1}(\lambda \Pi(K_r) + (1-\lambda) \Pi(K_t))^\frac{1}{n-1} \\
    \overset{\text{B.M.}}&{\geq} \lambda\vol_{n-1}(\Pi(K_r))^{\frac{1}{n-1}} + (1-\lambda)\vol_{n-1}(\Pi(K_t))^{\frac{1}{n-1}} \\
    &= \lambda\vol_{n-1}(K_r)^{\frac{1}{n-1}} + (1-\lambda)\vol_{n-1}(K_t)^{\frac{1}{n-1}}  \\
    &= \lambda f_{\theta}(r) + (1-\lambda)f_{\theta}(t).
\end{align*}
Continuity of $f_{\theta}(\cdot)$ can be shown as follows. Let $a=a(\theta)$ and $b=b(\theta)$ be as in (\ref{diretion_conditional_density}). By definition of $a$ we know that $T_{\theta,a}$ intersects $K$ only through the boundary of $K$. This intersection only contains a single point, if another point were in the intersection this would imply that the boundary of $K$ contains a line segment which contradicts the strict convexity of $K$. As a result: $f_{\theta}(a) = 0$ and similarly: $f_\theta(b) = 0$. Because $a$ and $b$ are the only possible points of discontinuity, $f_{\theta}(\cdot)$ is continuous.
\end{proof}

Because a bounded concave function has a maximum, strict concavity then implies that the maximum is unique. We obtain as a direct consequence of Theorem \ref{strictly_convex_result} and Theorem \ref{strict_concavity_condition}:
\begin{corollary}\label{strictly_convex_result}
Let $K \in \mathcal{K}^n$ strictly convex, and let $G_K$ be its section volume CDF. Then, $G_K$ is absolutely continuous.
\end{corollary}

Let us now consider approximation of convex bodies which are not necessarily strictly convex. We show that for any $K \in \mathcal{K}^n$ the CDF $G_K$ can be approximated arbitrarily closely by the CDF $G_L$ for some strictly convex $L \in \mathcal{K}^n$. This is due to the fact that any convex body may be approximated by a smooth and strictly convex body. A quantitatively useful statement is the following, see Theorem 1.5 in \cite{Klee1959}:

\begin{lemma}\label{strictly_convex_approximation}
Let $K \in \mathcal{K}^n$ a convex body with $0 \in \interior K$. Let $0 < \lambda < 1$. There exists a smooth and strictly convex body $L \in \mathcal{K}^n$ such that:
\begin{equation*}
    \lambda K \subset L \subset K. 
\end{equation*}
\end{lemma}

\begin{theorem}
Given a convex body $K \in \mathcal{K}^n$, there exists a sequence of strictly convex bodies $K_m\in \mathcal{K}^n$, $m \in \NN$, such that the sequence of the corresponding section volume CDFs $(G_{K_m})_{m \in \NN}$ converges pointwise to $G_K$ as $m \to \infty$. 
\end{theorem}

\begin{proof}
Let $K \in \mathcal{K}^n$ and let $G_K$ be its section volume CDF. By property 1 of Lemma \ref{basic_properties} we may assume without loss of generality that $0 \in \interior K$. Set $\lambda_m := 1 - 1/(m+1)$ for $m \in \NN$. Then $0 < \lambda_m < 1$ and $\lambda_m$ increases to 1 as $m \to \infty$. Using Lemma \ref{strictly_convex_approximation}, for $\lambda_m$ let $K_m$ be a smooth and strictly convex body such that: $\lambda_m K \subset K_m \subset K$. Let $z \in \RR$, by property 4 of Lemma \ref{basic_properties} we find:
\begin{align}
        G_K(z)\frac{\bar{b}(K)}{\bar{b}(K_m)} - \left(1 - \frac{\bar{b}(K)}{\bar{b}(K_m)}\right) \leq G_{K_m}(z) \leq G_{\lambda_m K}(z)\frac{\bar{b}(\lambda_m K)}{\bar{b}(K_m)} + \left(1 - \frac{\bar{b}(\lambda_m K)}{\bar{b}(K_m)} \right).
  \label{continuity_proof_bound}
\end{align}
By property 3 of Lemma \ref{basic_properties} we have $G_{\lambda_m K}(z) = G_K\left(z/\lambda_m^{n-1} \right)$. Note that $z / \lambda_m^{n-1}$ decreases towards $z$ as $m \to \infty$. Because $G_K$ is a CDF it is right-continuous, therefore:
\begin{equation}
    \lim_{m \to \infty} G_{\lambda_m K}(z) = \lim_{m \to \infty}G_K\left(\frac{z}{\lambda_m^{n-1}} \right) = G_K(z). \label{scaled_section_cdf_limit}
\end{equation}
Note that: $\lambda_m \bar{b}(K) = \bar{b}(\lambda_m K) \leq \bar{b}(K_m) \leq \bar{b}(K)$. As a result: 
\begin{equation}
    \lim_{m \to \infty} \bar{b}(K_m) = \bar{b}(K). \label{mean_width_limit}
\end{equation}
Combining (\ref{continuity_proof_bound}) with (\ref{scaled_section_cdf_limit}) and (\ref{mean_width_limit}), we obtain $\lim_{m \to \infty} G_{K_m}(z) = G_K(z)$.
\end{proof}

\subsection{Polytopes}

In this section we study polytopes, which are especially of interest for practical applications. Being examples of non-strictly convex bodies, they are not covered by Corollary \ref{strictly_convex_result}. The main result we obtain in this section is that the section volume CDF of a full-dimensional convex polytope is absolutely continuous. In order to obtain this result for polytopes, we need to deal with the regions where the function $f_{\theta}$, as in Brunn's theorem, is constant. The following lemma shows that this can only happen if the polytope has parallel edges.

\begin{lemma}\label{lemma_polytope_parallel_edges}
    Let $P \subset \RR^n$ be a full-dimensional convex polytope, $n \geq 2$. Fix $\theta \in S_{+}^{n-1}$ and define the function $f_{\theta} : \RR \to [0, \infty)$ by:
\[f_{\theta}(s) = \vol_{n-1}(P \cap T_{\theta, s})^{\frac{1}{n-1}}.\]
Suppose $f_\theta$ attains its maximum on the entire interval $[s_{-}, s_{+}]$, with $s_{-}< s_{+}$. Then, any plane $T_{\theta, s}$ with $s \in [s_{-}, s_{+}]$ intersects the same edges of $P$ and these edges are parallel.
\end{lemma}

\begin{proof}
    Let $a = a(\theta)$ and $b=b(\theta)$ as in (\ref{diretion_conditional_density}). For $s \in (a, b)$ the intersection $P\cap T_{\theta,s}$ is an $(n-1)$-dimensional polytope, and its vertices are the intersections of $T_{\theta,s}$ with the edges of $P$. By Brunn's theorem we know that $f_{\theta}$ is concave on its support. The set of points at which a concave function attains its maximum is convex, hence it is a nondegenerate interval or a single point. By assumption it is the interval $[s_{-}, s_{+}]$. By Brunn-Minkowski's inequality, and in particular its equality condition, we know that all sections $\{P\cap T_{\theta, s}: s \in [s_{-}, s_{+}]\}$ are equal up to dilatation and translation. But, because all such sections have equal volume, these sections then have to be equal up to translations. Write $P_{s_{-}} = P \cap T_{\theta, s_{-}}$ and $P_{s_{+}} = P \cap T_{\theta, s_{+}}$. Because $P_{s_{+}}$ is equal to $P_{s_{-}}$ up to translation there exists a $x \in \RR^n$ such that $P_{s_{+}} = P_{s_{-}} + x$. Let $s \in [s_{-},s_{+}]$, we claim that:
    \begin{equation}
        P \cap T_{\theta, s} = P_{s_{-}} + \frac{s-s_{-}}{s_{+}-s_{-}}x =: Q(s). \label{polytope_parallel_edges_claim}
    \end{equation}
    Let $z \in Q(s)$, then there exists a $y \in P_{s_{-}}$ such that:
    \begin{align*}
        z = y + \frac{s-s_{-}}{s_{+}-s_{-}}x = \frac{s_{+}-s}{s_{+}-s_{-}}y + \left(1 -  \frac{s_{+}-s}{s_{+}-s_{-}}\right)(x+y).
    \end{align*}
    Since $y \in P_{s_{-}}$ and $(x+y) \in P_{s_{-}} + x = P_{s_{+}}$, it follows that $z$ is the convex combination of two points in $P$, hence $z\in P$. Moreover, we have $\langle y,\theta \rangle = s_{-}$ and $\langle x+y,\theta \rangle = s_{+}$. A direct computation yields: $\langle z,\theta\rangle = s$. This means that $Q(s) \subset P \cap T_{\theta,s}$. Because $Q(s)$ is a translation of $P_{s_{-}}$ and since $P \cap T_{\theta,s}$ is equal to $P_{s_{-}}$ up to a translation we necessarily have that (\ref{polytope_parallel_edges_claim}) holds. Therefore, for any vertex $v$ of $P_{s_{-}}$, $v + ((s-s_{-})/(s_{+}-s_{-}))x$ is a vertex of $P \cap T_{\theta, s}$. It is evident that all vertices of the polytopes $\{P\cap T_{\theta, s}: s \in [s_{-}, s_{+}]\}$ lie on parallel line segments which are subsets of the edges of $P$, this finishes the proof. 
\end{proof}

In the next theorem we combine some of the techniques used earlier in this paper and Lemma \ref{lemma_polytope_parallel_edges} to show that the section volume CDF of any full-dimensional convex polytope is absolutely continuous.

\begin{theorem}
Let $P \subset \RR^n$ be a full-dimensional convex polytope, $n \geq 2$. Let $G_P$ be its section volume CDF. Then, $G_P$ is absolutely continuous.
\end{theorem}

\begin{proof}
    Given $\theta \in S_{+}^{n-1}$ define the function $f_{\theta} : \RR \to [0, \infty)$ by:
\[f_{\theta}(s) = \vol_{n-1}(P \cap T_{\theta, s})^{\frac{1}{n-1}}.\]
Let $B \subset \RR$ be a Borel set of Lebesgue measure zero. Let $f_P$ be as in Definition \ref{section_cdf_def} and let $(\Theta,S) \sim f_P$. As in the proof of Theorem \ref{strict_concavity_condition}, we condition on $\Theta=\theta$ and write:
\begin{align*}
\PP\Big(\vol_{n-1}(P\cap T_{\Theta, S})^{\frac{1}{n-1}} \in B\Big) = \PP\left(f_{\Theta}(S) \in B \right) = \int_{S_{+}^{n-1}}\PP\left(f_{\theta}(S) \in B \Big| \Theta = \theta \right)f_{P,\Theta}(\theta)\mathrm{d}\theta,
\end{align*}
with $f_{P,\Theta}(\theta)$ being the marginal density of $\Theta$ as in (\ref{direction_marginal_density}). In order to show that $G_P$ is absolutely continuous it is sufficient to show that $\PP\left(f_{\Theta}(S) \in B \right) = 0$. Let $a=a(\theta)$ and $b=b(\theta)$ be as in (\ref{diretion_conditional_density}) such that $S|\Theta = \theta \sim \mathcal{U}(a,b)$. Note that $f_{\theta}$ is continuous on $\RR$ for almost all $\theta \in S_{+}^{n-1}$. For almost all $\theta \in S_{+}^{n-1}$ the section planes $T_{\theta, s}$ enter the polytope through a vertex as $s$ runs from $a(\theta)$ to $b(\theta)$. For any such $\theta$, $f_{\theta}(a) = 0$ and $f_{\theta}(b) = 0$, because a vertex has no $(n-1)$-dimensional volume. As $a$ and $b$ are the only possible points of discontinuity, $f_{\theta}$ is continuous on $\RR$ for almost all $\theta \in S_{+}^{n-1}$. 

By Brunn's theorem we know that $f_{\theta}$ is concave on its support. The set of points at which a concave function attains its maximum is convex, hence it is a nondegenerate interval or a single point. Denote this set by: $[s_{-}(\theta), s_{+}(\theta)]$, in the case it consists of a single point, $s_{-}(\theta) = s_{+}(\theta)$. Write: $p_1 = \PP(S \in (a, s_{-})|\Theta =\theta)$, $p_2 = \PP(S \in [s_{-}, s_{+}]|\Theta =\theta)$ and $p_3 = \PP(S \in (s_{+}, b)|\Theta =\theta)$. We may write:
\begin{align*}
\begin{split}
    \PP(f_{\theta}(S) \in B|\Theta = \theta) ={}& \PP(f_{\theta}(S) \in B|\Theta=\theta, S\in(a,s_{-}))p_1 + \PP(f_{\theta}(S) \in B|\Theta=\theta, S\in[s_{-},s_{+}])p_2 + \\
    & + \PP(f_{\theta}(S) \in B|\Theta=\theta, S\in(s_{+},b))p_3. \\
\end{split}
\end{align*}
Arguing as in the proof of Theorem \ref{strict_concavity_condition} we obtain that for almost all $\theta \in S_{+}^{n-1}$ the distribution functions $z \mapsto \PP(f_{\theta}(S) \leq z|\Theta=\theta, S\in(a,s_{-}))$ and $z \mapsto \PP(f_{\theta}(S) \leq z|\Theta=\theta, S\in(s_{+},b))$ are continuous and convex on their support and therefore absolutely continuous with respect to Lebesgue measure. Hence, for any such $\theta$ we have $\PP(f_{\theta}(S) \in B|\Theta=\theta, S\in(a,s_{-})) = 0$ and $\PP(f_{\theta}(S) \in B|\Theta=\theta, S\in(s_{+},b))= 0$. Clearly, for any $\theta \in S_{+}^{n-1}$: $p_2(\theta) = (s_{+}(\theta) - s_{-}(\theta))/(b(\theta) - a(\theta))$. Further note that $\PP(f_{\theta}(S) \in B|\Theta=\theta, S\in[s_{-},s_{+}]) = \mathds{1}{\{m_P(\theta)^{1/(n-1)} \in B\}}$, with $m_P(\cdot)$ is as in (\ref{inner_section_function}) and $\theta \in S_{+}^{n-1}$. Combining these results we may therefore write:
\begin{align}
    \PP(f_{\Theta}(S) \in B) = \int_{S_{+}^{n-1}}\mathds{1}{\{m_P(\theta)^{\frac{1}{n-1}} \in B\}}\frac{s_{+}(\theta) - s_{-}(\theta)}{b(\theta) - a(\theta)}f_{P,\Theta}(\theta)\mathrm{d}\theta. \label{polytope_proof_conditioning}
\end{align}
In (\ref{polytope_proof_conditioning}) we effectively only integrate over $\theta$ such that $s_{+}(\theta) > s_{-}(\theta)$. By Lemma \ref{lemma_polytope_parallel_edges} this strict inequality only holds if for all $s \in [s_{-}(\theta),  s_{+}(\theta)]$ the same edges of $P$ are intersected by $T_{\theta,s}$ and these edges are all parallel. Define:
\[D = \left\{\theta \in S_{+}^{n-1}: s_{+}(\theta) > s_{-}(\theta)\right\}.\]
Hence, for any $\theta \in D$ and any $s \in [s_{-}(\theta),  s_{+}(\theta)]$, we have $m_P(\theta) = \vol_{n-1}(P \cap T_{\theta, s})$ and the plane $T_{\theta, s}$ only intersects edges of $P$ which are parallel. If $\sigma_{n-1}(D) = 0$, for example because $P$ does not have any parallel edges, then $\PP(f_{\Theta}(S) \in B) = 0$, hence $G_P$ is absolutely continuous. 

Let us now consider the case $\sigma_{n-1}(D) > 0$. We may write $D$ as a disjoint union $D = \cup_{i=1}^k D_i$ for some $k \in \NN$. Here $D_i$ is defined such that for all $\theta \in D_i$ all planes corresponding to $m_P(\theta)$ intersect the same parallel edges. Let $i \in \{1,\dots,k\}$ and let $e_1,\dots,e_m \subset P$ be the parallel edges of $P$ corresponding to $D_i$. Consider the plane $T_{\phi_i,0}$, with $\phi_i \in S_{+}^{n-1}$ such that this plane is orthogonal to the edges $e_1,\dots,e_m$. For any $L \subset \RR^n$ let $\Pi_i(L)$ denote the orthogonal projection of $L$ on the hyperplane $T_{\phi_i,0}$. Let $\theta \in D_i$, $s \in  [s_{-}(\theta),  s_{+}(\theta)]$ such that $m_P(\theta) = \vol_{n-1}(P \cap T_{\theta, s})$ and the plane $T_{\theta, s}$ intersects $e_1,\dots,e_m$. Note that for any $\theta \in D_i$ and $s \in  [s_{-}(\theta),  s_{+}(\theta)]$, $v_i :=\vol_{n-1}(\Pi_i(P\cap T_{\theta, s}))$ attains the same value. After all, for any such plane, $\Pi_i(P\cap T_{\theta, s})$ is a polytope in $T_{\phi_i, 0}$ and its vertices are given by the orthogonal projections of $e_1,\dots,e_m$ on $T_{\phi_i, 0}$. Moreover, it is well known that the volume of $P \cap T_{\theta, s}$ and the volume of its projection on $T_{\phi_i, 0}$ are related via:
\[v_i =\vol_{n-1}(\Pi_i(P\cap T_{\theta, s})) = |\langle \theta, \phi_i \rangle|\vol_{n-1}(P \cap T_{\theta, s}).\]
Hence,
\begin{equation}
    m_P(\theta) = \vol_{n-1}(P \cap T_{\theta, s}) = \frac{v_i}{|\langle \theta, \phi_i \rangle|}, \text{\quad} \theta \in D_i. \label{polytope_sectionfunction_local}
\end{equation}
If we were to draw $\Omega \sim \mathcal{U}(S^{n-1})$, then the Lebesgue density of the random variable $\langle \Omega,\phi_i\rangle$ is given by:
\[t \in [-1, 1] \mapsto \frac{\Gamma(\frac{n}{2})}{\sqrt{\pi}\Gamma(\frac{n-1}{2})}(1-t^2)^{\frac{n-3}{2}}.\]
This density does not depend on $\phi_i$ due to symmetry. Because the probability measure corresponding to the uniform distribution on the sphere is the normalized spherical measure, we obtain $\sigma_{n-1}(\{\theta \in S^{n-1}: \langle \theta, \phi_i \rangle \in B\}) = 0$. Via the change of variables formula it is easily verified that the random variable $(v_i/|\langle \Omega, \phi \rangle|)^{1/(n-1)}$ also has a Lebesgue density. Therefore: 
\begin{equation}
    \sigma_{n-1}\left(\left\{\theta \in S^{n-1}: \left(v_i /|\langle \theta, \phi_i \rangle|\right)^{\frac{1}{n-1}} \in B\right\}\right) = 0. \label{null_set_sphericalrv}
\end{equation}
Finally, from its definition it is evident that the density $f_{P,\Theta}$ is bounded, see (\ref{direction_marginal_density}). Let $M > 0$ be an upper bound of this density. Using this fact and (\ref{polytope_sectionfunction_local}) and (\ref{null_set_sphericalrv}), the claim follows:
\begin{align*}
    \PP(f_{\Theta}(S) \in B) &= \int_{S_{+}^{n-1}}\mathds{1}{\left\{m_P(\theta)^{\frac{1}{n-1}} \in B\right\}}\frac{s_{+}(\theta) - s_{-}(\theta)}{b(\theta) - a(\theta)}f_{P,\Theta}(\theta)\mathrm{d}\theta \\
    &\leq \int_D \mathds{1}{\left\{m_P(\theta)^{\frac{1}{n-1}} \in B\right\}}f_{P,\Theta}(\theta)\mathrm{d}\theta \\
    &\leq M \sum_{i=1}^k \int_{D_i} \mathds{1}{\left\{m_P(\theta)^{\frac{1}{n-1}} \in B\right\}}\mathrm{d}\theta \\
    &= M \sum_{i=1}^k \int_{D_i} \mathds{1}{\left\{\left(v_i /|\langle \theta, \phi_i \rangle|\right)^{\frac{1}{n-1}} \in B\right\}}\mathrm{d}\theta \\
    &\leq M \sum_{i=1}^k \sigma_{n-1}\left(\left\{\theta \in S^{n-1}: \left(v_i /|\langle \theta, \phi_i \rangle|\right)^{\frac{1}{n-1}} \in B\right\}\right)\\
    &=0.
\end{align*}
\end{proof}

\section{Density approximation}

In this section we consider convex bodies $K \in \mathcal{K}^n$ such that $G_K$ is absolutely continuous. For most convex bodies $K \in \mathcal{K}^n$ there is no known explicit expression for $G_K$ or its density $g_K$. In this section we focus on approximating the density $g_K$. This is achieved by obtaining a large sample from $G_K$ along with a kernel density estimator (KDE). We use the following rejection sampling scheme, proposed in \cite{Miles1978}, to sample from the distribution $G_K$:

\begin{enumerate}
    \item Enclose $K$ inside a sphere: choose $R > 0$ such that $K \subset \bar{B}(0, R)$.
    \item Choose an isotropic random direction $\Theta \sim \mathcal{U}(S^{n-1})$.
    \item Sample $S \sim \mathcal{U}(0,R)$.
    \item The plane $T_{\Theta,S}$ hits $\bar{B}(0,R)$, if the plane also hits $K$ we accept, and $Z = \vol_{n-1}(K\cap T_{\Theta, S})$ is a draw from $G_K$. If the plane does not hit $K$, we reject and go back to step 2. 
\end{enumerate}

In $\RR^2$, step 2 may be achieved by sampling $\Phi \sim \mathcal{U}(0,2\pi)$ followed by setting $\Theta = (\cos(\Phi), \sin(\Phi))$. In $\RR^3$ step 2 can be performed as follows. Sample $\Phi \sim \mathcal{U}(0,2\pi)$ and $X \sim \mathcal{U}(-1,1)$. Then, we may set $\Omega = \arccos\left(X\right)$ and $\Theta = (\sin\Omega\cos\Phi, \sin\Omega\sin\Phi,\cos\Omega)$. In order to keep the rejection rate in the sampling scheme low, $R$ should be as small as possible and $K$ should be positioned at the origin, meaning that $0 \in \interior K$.

Of course, for $K \in \mathcal{K}^n$ with $n=2$ we find: $G_K \equiv G_K^S$. Note that $G_K^S$ is initially convex on some initial interval (see Remark \ref{convexity_hk_remark}). As a result, its density $g_K^S$ is non-decreasing on this interval. This means that $g_K^S$ may even be constant initially. In addition, note that if $n=2$ and $K$ is a polygon then in \cite{Gates1987} it has been shown that $g_K^S$ is always constant on some initial interval. Because of the relation between $g_K^S$ and $g_K$, if $g_K^S$ is constant on an initial interval, then $g_K$ behaves like $z^{(2-n)/(n-1)}$ on this interval. Hence, when $n=3$ this means that $g_K$ behaves like $1/\sqrt{z}$ for $z$ close to zero. Clearly, this complicates the approximation of $g_K$ near zero. Therefore, we choose to approximate the density $g_K^S$ instead, and use (\ref{g_and_h_relation}) to obtain an approximation of $g_K$.

We will now introduce the Monte Carlo simulation scheme for approximating $g_K^S$. We choose a large $N\in \NN$ and sample $Z_1,\dots,Z_N \overset{\mathrm{iid}}{\sim} G_K$ using the sampling scheme given above. Setting $X_i = Z_i^{1/(n-1)}$, we obtain that $X_1,\dots,X_N \overset{\mathrm{iid}}{\sim} G_K^S$. The following KDE is for example studied in \cite{Schuster1985}, which we propose as an approximation for $g_K^S$:
\begin{equation}
\hat{g}_N^S(z) = \frac{1}{hN}\sum_{i=1}^N k\left(\frac{z-X_i}{h}\right) + k\left(\frac{z+X_i}{h} \right), \text{ \quad } z \geq 0, \label{reflection_method}
\end{equation}
with $h > 0$ the bandwidth parameter and $k$ a symmetric kernel. The KDE in (\ref{reflection_method}) is also known as the reflection method. A reason for using the reflection method over the classical (Parzen-Rosenblatt) KDE is that it ensures that no probability mass is assigned for $z < 0$. Recall that the classical KDE for the sample $X_1,\dots,X_N$ is given by:
\begin{equation}
f(z) = \frac{1}{hN}\sum_{i=1}^n k\left(\frac{z - X_i}{h}\right), \text{ \quad } z \in \RR. \label{classical_kde}
\end{equation}
Note that when computing the KDE in (\ref{classical_kde}) for the following 'sample' of size $2N$:
\[X_1,X_2,\dots,X_N,-X_1,-X_2,\dots,-X_N,\]
we find: $f(z) = \hat{g}_N^S(z)/2$. This fact may be used to choose the bandwidth $h$, since most of the literature is devoted to bandwidth selection for the classical KDE. In the data examples in the next section we choose for $k$ the Gaussian kernel and select the bandwidth with the popular Sheather-Jones method \cite{Sheather1991}. Whenever we want to approximate $g_K$ instead of $g_K^S$, we simply follow the procedure given above to compute $\hat{g}_N^S$. Then, using (\ref{g_and_h_relation}) we set:
\begin{equation}
\hat{g}_N(z) = \hat{g}_N^S\left(z^{\frac{1}{n-1}}\right)\frac{z^{\frac{2-n}{n-1}}}{n-1},
\end{equation}
which is an approximation of $g_K$. A drawback of the KDE in (\ref{reflection_method}) is that this density has (right)-derivative zero in $z=0$. As mentioned before, when $n=2$ and $K$ is a convex polygon this is not an issue since the density $g_K^S$ is then initially constant. In the data examples in the next section the approximations of $G_P^S$ of some convex polytopes $P$ in $\RR^3$ appear initially (close to) linear. This suggests that the choice of boundary correction is reasonable. Should one consider a polytope $P$ such that $G_P^S$ is far from being initially linear then other boundary correction methods may be more appropriate.

\begin{figure}[t!]
\begin{center}
\includegraphics[scale=0.84]{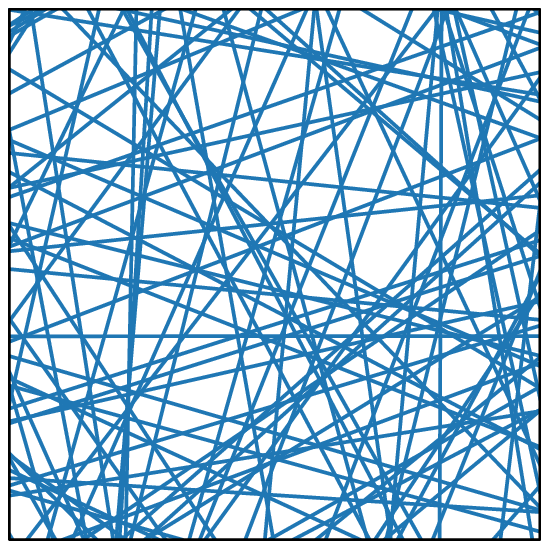}\includegraphics[scale=0.84]{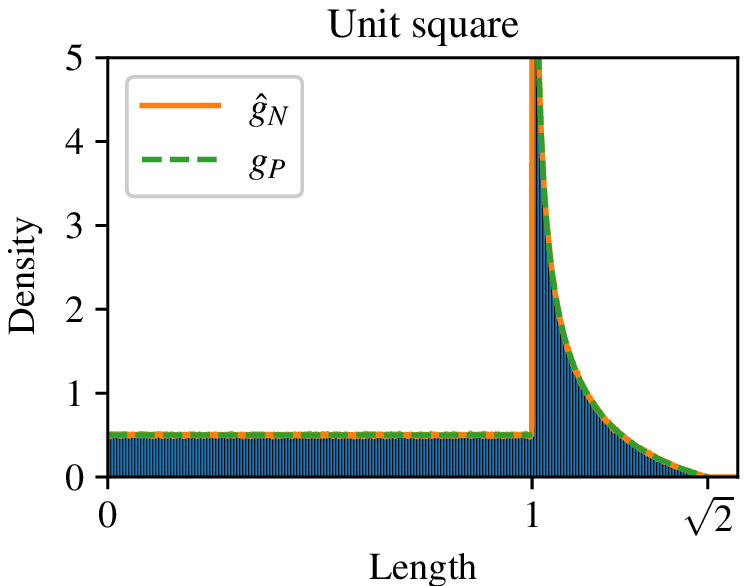}
\caption{Left: 100 IUR sections through $P$ the unit square in $\RR^2$. Right: comparison of the density $g_P$ to its approximation $\hat{g}_N$.}\label{unit_square_simulations}
\end{center}
\end{figure}

\begin{figure}[t!]
\begin{center}
\includegraphics[scale=0.84]{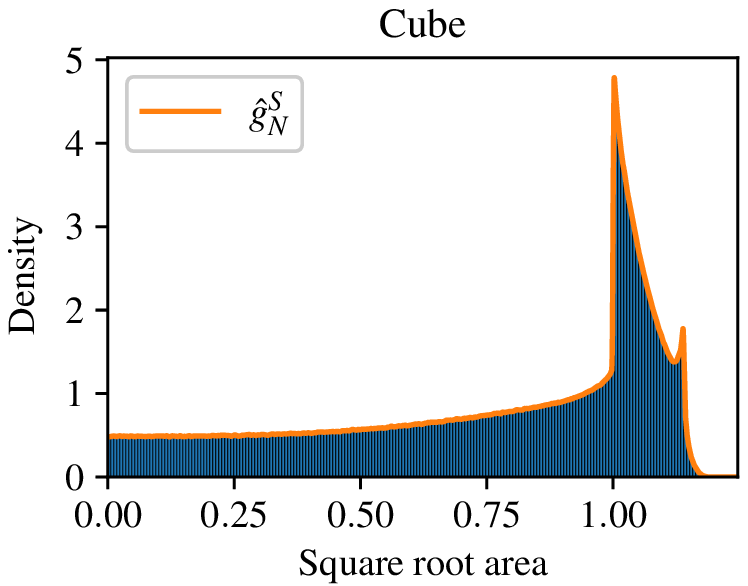}\includegraphics[scale=0.84]{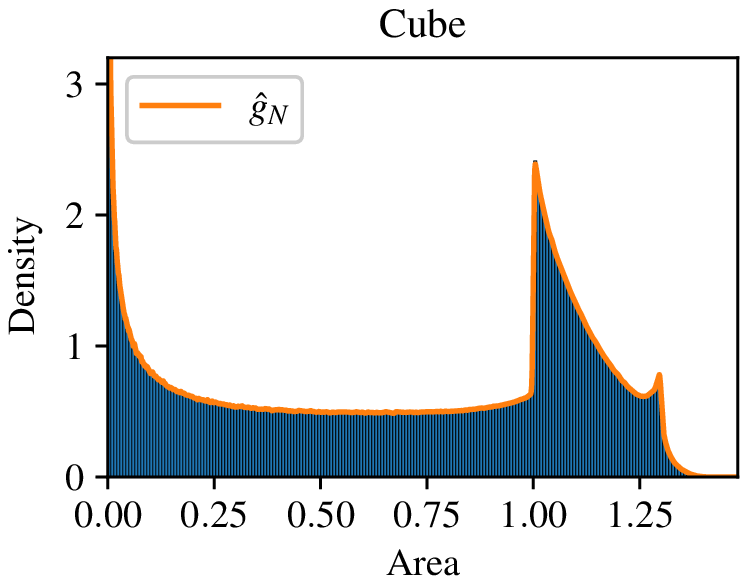}
\includegraphics[scale=0.84]{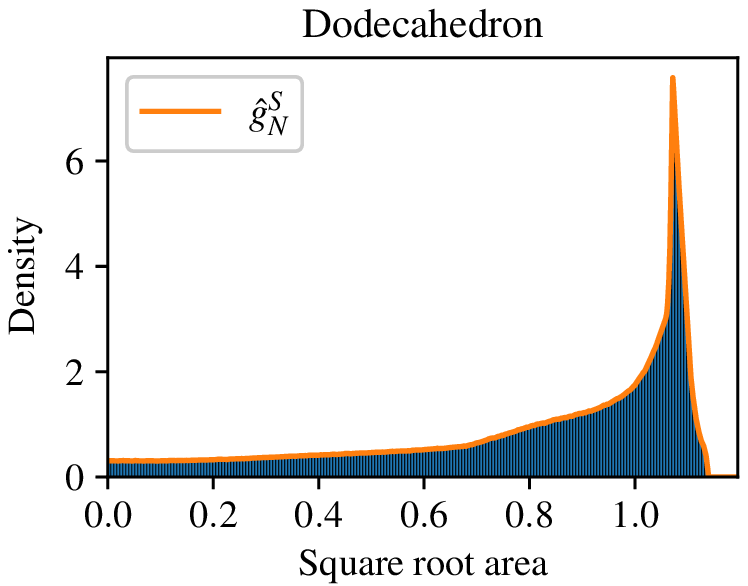}\includegraphics[scale=0.84]{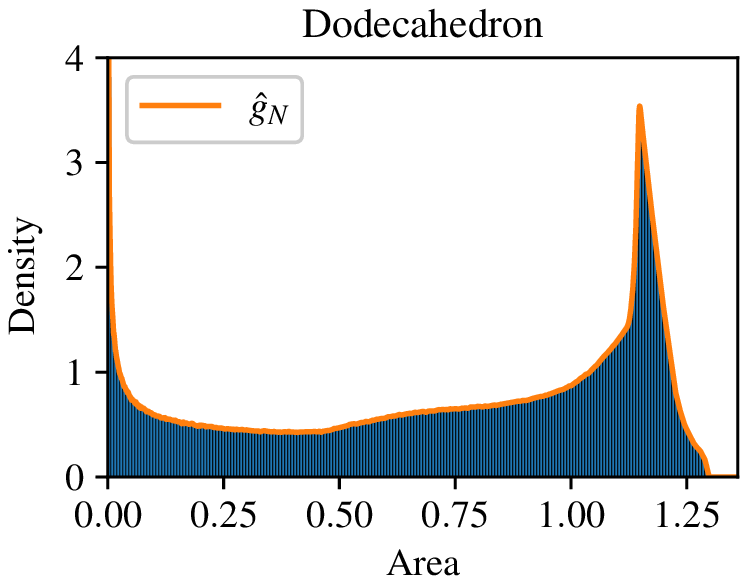}
\caption{Approximations of $g_P^S$ and $g_P$ for $P$ the unit cube (upper panel) and for $P$ the dodecahedron with volume 1 (lower panel).}\label{cube_simulations}
\end{center}
\end{figure}

\subsection{Simulations}
In this section we perform a few simulations to show that the Monte Carlo simulation scheme works well. For these simulations we focus on polytopes. Throughout this section, let $P\subset \RR^n$ be a full-dimensional polytope. We have implemented the sampling scheme for drawing samples from $G_P$ specifically for $n=2$ and $n=3$. The code used for the simulations may be found at \url{https://github.com/thomasvdj/pysizeunfolder}. The polytope can be entered into this program either by presenting a set of points, such that the polytope is given by the convex hull of these points, or by presenting a half-space representation of the polytope. 

In the literature, similar simulations have been performed, e.g. for the cube and the dodecahedron. Therefore we also consider these shapes, such that we have a point of comparison. Besides approximating $g_P$ and $g_P^S$ we also approximate $G_P^S$. The distribution function $G_P^S$ can be approximated arbitrarily closely by an empirical distribution function, given a large sample from $G_P^S$.  

For all simulations, we set $N=10^7$. For the first example, we choose the unit square in $\RR^2$. The density of its chord length distribution may be found in \cite{Coleman1969}, it is given by:
\[g_P(z) = \begin{cases}\frac{1}{2} & \text{if } 0 \leq z \leq 1 \\
    \frac{1}{z^2\sqrt{z^2 - 1}} - \frac{1}{2} & \text{if } 1 < z \leq \sqrt{2} 
    \end{cases}.\]
The approximation obtained via the proposed Monte Carlo scheme is shown in Figure \ref{unit_square_simulations}. Figure \ref{unit_square_simulations} also contains a visualization of 100 IUR sections through the unit square. As can be seen in Figure \ref{unit_square_simulations}, the approximation $\hat{g}_N$ is very close to the true probability density $g_P$. 

\begin{figure}[b!]
\begin{center}
\includegraphics[scale=0.84]{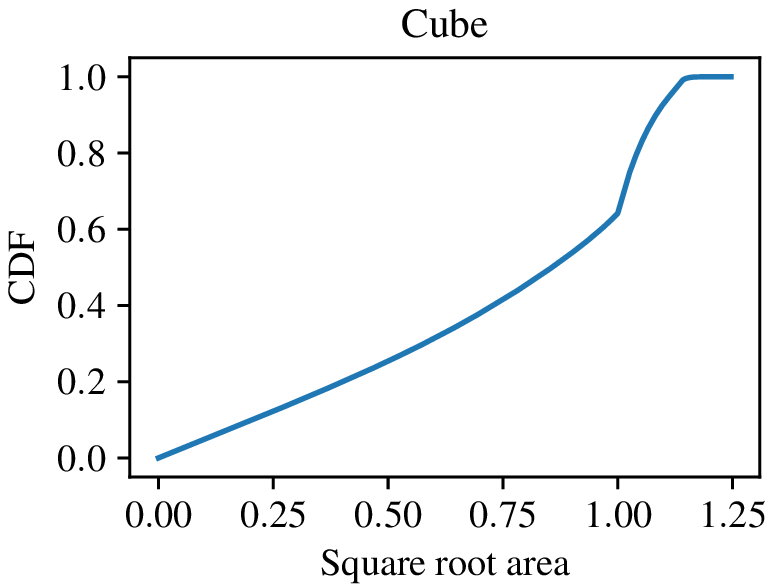}\includegraphics[scale=0.84]{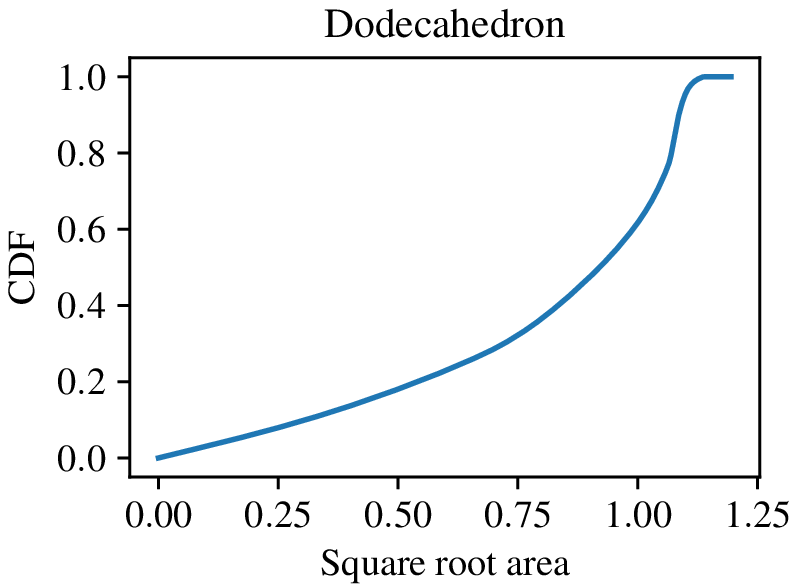}
\caption{Approximations of $G_P^S$ for $P$ the unit cube (Left) and for $P$ the dodecahedron with volume 1 (Right).}\label{cdf_estimates}
\end{center}
\end{figure}

We should stress that the proposed method is especially useful in the spatial setting $n=3$. Naturally, whenever the analytical expression for $g_P$ is available this is preferable. To the best of our knowledge, there are no known expressions for $g_P$ of any polytope $P$ in $\RR^3$. In the planar case ($n=2$) the density $g_P$ is known for various polygons, for example for rectangles \cite{Coleman1969}, and regular polygons \cite{Harutyunyan2009}. In Figure \ref{cube_simulations} the approximations of $g_P$ and $g_P^S$ are shown for the cube and the dodecahedron, both shapes scaled to have volume 1. 

Similar simulations were performed in \cite{Paul1981} for the cube and dodecahedron, qualitatively the curves visualized there are close to the approximations of $g_P$ shown in Figure \ref{cube_simulations}. For the cube, one can easily see that for any direction $\theta \in S_{+}^{n-1}$, there exists a section of area 1. By Remark \ref{convexity_hk_remark}, this means that $g_P^S$ is non-decreasing on $(0, 1)$, which can also be seen in Figure \ref{cube_simulations}. Approximations of $G_P^S$ for the cube and dodecahedron are shown in Figure \ref{cdf_estimates}. For these visualizations the same samples are used as in Figure \ref{cube_simulations}. As mentioned before, these approximations of $G_P^S$ appear initially (close to) linear, justifying the choice of boundary correction in the density approximation procedure.

\section{Concluding remarks}
In this paper we establish absolute continuity of the (transformed) section volume CDF for various classes of convex bodies. Absolute continuity of these distribution functions is essential for stereological estimation of particle size distributions using likelihood-based inference. Whether these distribution functions are absolutely continuous for all convex bodies remains an open problem. From a theoretical perspective we cover a large class of convex bodies with the strictly convex bodies. With polytopes we cover a class of convex bodies which is especially important in practical applications. Moreover, for polytopes we provide a Monte Carlo simulation scheme for approximating the density corresponding to its (transformed) section volume CDF. 

\appendix
\section{Appendix: additional proofs}\label{appendix_proofs}

\begin{proof}[Proof of Lemma \ref{basic_properties}]
Let $x \in \RR^n$, $\theta \in S_{+}^{n-1}$ and $s \in \RR$. It can be easily verified that the following holds:
\begin{equation}
    (K+x) \cap T_{\theta, s} = \left(K \cap T_{\theta, s - \langle x,\theta\rangle}\right) + x. \label{translation_equality}
\end{equation}
Meaning that the intersection of a translated $K$ with a plane is the same as the intersection of $K$ with a translated plane and then translating the result. It follows that:
\begin{align}
G_{K+x}(z) &= \int_{S_{+}^{n-1}}\int_{\RR} \frac{\mathds{1}{\{\vol_{n-1}((K+x)\cap T_{\theta, s}) \leq z\}}\mathds{1}{\{(K+x) \cap T_{\theta, s} \neq \emptyset\}}}{\mu([K+x])}\mathrm{d}s\mathrm{d}\theta \nonumber\\
\overset{(\ref{translation_equality})}&{=} \int_{S_{+}^{n-1}}\int_{\RR} \frac{\mathds{1}{\{\vol_{n-1}((K \cap T_{\theta, s - \langle x,\theta\rangle}) + x) \leq z\}}\mathds{1}{\{K \cap T_{\theta, s - \langle x,\theta\rangle } \neq \emptyset\}}}{\mu([K+x])}\mathrm{d}s\mathrm{d}\theta \nonumber\\
&= \int_{S_{+}^{n-1}}\int_{\RR} \frac{\mathds{1}{\{\vol_{n-1}(K \cap T_{\theta, s - \langle x,\theta\rangle}) \leq z\}}\mathds{1}{\{K \cap T_{\theta, s - \langle x,\theta\rangle} \neq \emptyset\}}}{\mu([K+x])}\mathrm{d}s\mathrm{d}\theta \label{lebesgue_translation_invariance}\\
&= \int_{S_{+}^{n-1}}\int_{\RR} \frac{\mathds{1}{\{\vol_{n-1}(K \cap T_{\theta, t}) \leq z\}}\mathds{1}{\{K \cap T_{\theta, t} \neq \emptyset\}}}{\mu([K+x])}\mathrm{d}t\mathrm{d}\theta \nonumber
\end{align}
In (\ref{lebesgue_translation_invariance}) we use the translation invariance of the Lebesgue measure. The final step is obtained by substituting $t = s - \langle x,\theta\rangle$. Via the same substitution it can be shown that $\mu([K+x]) = \mu([K])$. As a result we obtain: $G_{K+x}(z) = G_K(z)$. Moving on to the rotation invariance, let $M \in \SO(n)$, then the following can be shown:
\begin{equation}
    MK \cap T_{\theta, s} = M(K \cap T_{M^T \theta,s}). \label{rotating_equality}
\end{equation}
Using this, we find:
\begin{align}
G_{MK}(z) &= \frac{1}{2}\int_{S^{n-1}}\int_{\RR} \frac{\mathds{1}{\{\vol_{n-1}(MK\cap T_{\theta, s}) \leq z\}}\mathds{1}{\{MK \cap T_{\theta, s} \neq \emptyset\}}}{\mu([MK])}\mathrm{d}s\mathrm{d}\theta \label{hemi_sphere_substitution} \\
\overset{(\ref{rotating_equality})}&{=} \frac{1}{2}\int_{S^{n-1}}\int_{\RR} \frac{\mathds{1}{\{\vol_{n-1}(M(K \cap T_{M^T \theta,s})) \leq z\}}\mathds{1}{\{M(K \cap T_{M^T \theta,s}) \neq \emptyset\}}}{\mu([MK])}\mathrm{d}s\mathrm{d}\theta \nonumber \\
&= \frac{1}{2}\int_{S^{n-1}}\int_{\RR} \frac{\mathds{1}{\{\vol_{n-1}(K \cap T_{M^T \theta,s}) \leq z\}}\mathds{1}{\{K \cap T_{M^T \theta,s} \neq \emptyset\}}}{\mu([MK])}\mathrm{d}s\mathrm{d}\theta \label{lebesgue_rotation_invariance}  \\
&= \int_{S_{+}^{n-1}}\int_{\RR} \frac{\mathds{1}{\{\vol_{n-1}(K \cap T_{u,s}) \leq z\}}\mathds{1}{\{K \cap T_{u,s} \neq \emptyset\}}}{\mu([MK])}\mathrm{d}s\mathrm{d}u. \label{rotation_invariance_substitution}
\end{align}
In (\ref{hemi_sphere_substitution}) we use the fact that the inner integral does not change if we replace $\theta$ with $-\theta$, therefore we may integrate over $S^{n-1}$ instead and divide the result by two. In (\ref{lebesgue_rotation_invariance}) we use the rotation invariance of the Lebesgue measure. In (\ref{rotation_invariance_substitution}) the substitution: $u = M^T \theta$ is applied. Because $M$ is an orthogonal matrix of determinant one, the Jacobian corresponding to the transformation has determinant one. Since $M^T S^{n-1} = S^{n-1}$, the transformation does not affect the integration region. Via the same substitutions it can be shown that $\mu([MK]) = \mu([K])$ such that indeed $G_{MK}(z) = G_K(z)$. Next, we consider scaling of convex bodies. Let $\lambda > 0$, we remark that the following holds:
\begin{equation}
    \lambda K \cap T_{\theta,s} = \lambda\left(K \cap T_{\theta,\frac{s}{\lambda}}\right).
\end{equation}
Using this and the fact that $\vol_{n}(\lambda K) = \lambda^{n}\vol_{n}(K)$ for $K \in \mathcal{K}^n$, it is once again a matter of applying a substitution to obtain:
\begin{align*}
        G_{\lambda K}(z) = \lambda\int_{S_{+}^{n-1}}\int_{\RR} \frac{\mathds{1}{\{\vol_{n-1}(K \cap T_{\theta, t}) \leq \frac{z}{\lambda^{n-1}}\}}\mathds{1}{\{K \cap T_{\theta, t} \neq \emptyset\}}}{\mu([\lambda K])}\mathrm{d}t\mathrm{d}\theta.
\end{align*}
And similarly, via substitution we find: $\mu([\lambda K]) = \lambda \mu([K])$ such that indeed: $G_{\lambda K}(z) = G_K(z/\lambda^{n-1})$. We now consider the final statement of the lemma. Let $T$ be an IUR plane hitting $L$. By proposition 1 in \cite{Davy1977}, the probability that $T$ hits $K$ is given by $\bar{b}(K)/\bar{b}(L)$. Moreover, given that $T$ hits $K$ it is an IUR plane hitting $K$. It follows that:
    \begin{align*}
    G_L(z) &= \PP(\vol_{n-1}(L \cap T) \leq z) \\
        &= \PP(\vol_{n-1}(L \cap T) \leq z| K\cap T \neq \emptyset)\frac{\bar{b}(K)}{\bar{b}(L)} + \PP(\vol_{n-1}(L \cap T) \leq z| K\cap T = \emptyset)\left(1-\frac{\bar{b}(K)}{\bar{b}(L)}\right) \\
        &\leq \PP(\vol_{n-1}(K \cap T) \leq z| K\cap T \neq \emptyset)\frac{\bar{b}(K)}{\bar{b}(L)} + \left(1-\frac{\bar{b}(K)}{\bar{b}(L)}\right)\\
        &= G_K(z)\frac{\bar{b}(K)}{\bar{b}(L)} + \left(1 - \frac{\bar{b}(K)}{\bar{b}(L)} \right).
    \end{align*}
\end{proof}

\begin{proof}[Proof of Lemma \ref{strict_inclusion_lemma}]
Let $x \in \partial K$. Since $x \in \interior L$ there exists a $R > 0$ such that $B(x,R)\subset \interior L$. Because $x \in \partial K$ we know that $B(x,R) \cap (\RR^n \setminus K) \neq \emptyset$. Choose $y \in B(x,R) \cap (\RR^n \setminus K)$. Note that $y \in \interior L$ and $\interior L$ is open. Choose $r_1 > 0$ such that $B(y, r_1) \subset \interior L$. Because $K$ is closed, $\RR^n \setminus K$ is open. Choose $r_2 > 0$ such that $B(y,r_2) \subset \RR^n \setminus K$. Let $r = \min\{r_1, r_2\}$, then $B(y,r) \subset (\interior L) \setminus K$ and this ball has a strictly positive volume. Hence, we find:
\begin{align*}
    \vol_n (L) = \vol_n(L \setminus K) + \vol_n(K) \geq \vol_n(B(y,r)) + \vol_n(K) > \vol_n(K).
\end{align*}
\end{proof}

\section*{Acknowledgements}
This research was carried out under project number T17019s in the framework of the Research Program of the Materials innovation institute (M2i) (www.m2i.nl) supported by the Dutch government. We thank Kees Bos, Jilt Sietsma and Karo Sedighiani for fruitful discussions.

\bibliographystyle{ieeetr}
\bibliography{export}
\end{document}